\documentclass[journal, twoside, draftcls,onecolumn]{IEEEtran}
\ifCLASSINFOpdf
\else
\fi

\usepackage[small]{caption}
\usepackage{graphicx}
\usepackage{subfigure}
\usepackage{amsthm}
\usepackage{amsmath}
\usepackage{amssymb}
\usepackage{algorithmic}
\usepackage{algorithm}
\usepackage{color}

\theoremstyle{definition}
\newtheorem{ex}{Example}
\newtheorem{defi}{Definition}
\newtheorem{probdefi}{Problem}
\theoremstyle{Theorem}
\newtheorem{thrm}{Theorem}
\newtheorem{lemma}{Lemma}

\DeclareMathOperator*{\argmin}{arg\,min}
\begin{document}

\title{On the Construction of Data Aggregation Tree with Minimum Energy Cost
in Wireless Sensor Networks: NP-Completeness and Approximation Algorithms}

\author{Tung-Wei~Kuo,
        Kate~Ching-Ju~Lin,~\IEEEmembership{Member,~IEEE,}
        and~Ming-Jer~Tsai,~\IEEEmembership{Member,~IEEE}% <-this % stops a space
\thanks{Tung-Wei Kuo and Ming-Jer Tsai are with the Department Computer Science, National Tsing Hua University,
Hsinchu, Taiwan 30013, ROC.}% <-this % stops a space
\thanks{Tung-Wei kuo and Kate Ching-Ju Lin are with Research Center for
Information Technology Innovation, Academia Sinica, Taipei, Taiwan.}% <-this % stops a space
\thanks{This paper was presented in part at IEEE INFOCOM 2012.}% <-this % stops a space
}

%\markboth{IEEE TRANSACTIONS ON INFORMATION THEORY, VOL. xx, NO. xx, xxxxxx xxxx}%
%{Kuo \MakeLowercase{\textit{et al.}}: On the Construction of Data Aggregation Tree %with Minimum Energy Cost
%in Wireless Sensor Networks...}

%\author{
%\authorblockN{Tung-Wei Kuo$^\dag$ $^\ddag$, Kai-Hong Wang$^\ddag$, Kate Ching-Ju Lin$^\dag$, and Ming-Jer Tsai$^\ddag$}
%\authorblockA{
%$^\dag$Research Center for Information Technology Innovation, Academia Sinica, Taiwan\\
%$^\ddag$Department of Computer Science, National Tsing Hua University, Taiwan\\
%twkuo@citi.sinica.edu.tw, katelin@citi.sinica.edu.tw, and mjtsai@cs.nthu.edu.tw}
%}

\maketitle

\begin{abstract}
In many applications, it is a basic operation for the sink to periodically collect reports from
all sensors. Since the data gathering process usually proceeds for many rounds, it is important to
collect these data efficiently, that is, to reduce the energy cost of data transmission. Under
such applications, a tree is usually adopted as the routing structure to save the computation
costs for maintaining the routing tables of sensors. In this paper, we work on the problem of
constructing a data aggregation tree that minimizes the total energy cost of data transmission in
a wireless sensor network. In addition, we also address such a problem in the wireless sensor
network where relay nodes exist. We show these two problems are NP-complete, and propose
$O(1)$-approximation algorithms for each of them. Simulations show that the proposed algorithms
each have good performance in terms of the energy cost.
\end{abstract}


\IEEEpeerreviewmaketitle

\section{Introduction}
\IEEEPARstart{I}{n} many applications, sensors are required to send reports to a specific target (e.g. base
station) periodically~\cite{R1}. In habitat monitoring~\cite{R9} and civil structure maintenance~\cite{R10}, it is a basic operation for the sink to periodically collect reports from sensors.
Since the data gathering process usually proceeds for many rounds, it is necessary to reduce the
number of the packets, which carries the reports, transmitted in each round for energy saving. In
this paper, we undertake the development of data gathering in wireless sensor networks.

Data aggregation is a well-known method for data gathering, which can be performed in various
ways. In~\cite{R1}, a fixed number of reports received or generated by a sensor are aggregated
into one packet. In other applications, a sensor can aggregate the reports received or generated into 
one report using a divisible function 
(e.g. SUM, MAX, MIN, AVERAGE, top-k, etc.)~\cite{R11}. Data 
compression, which deals with the correlation between data such that the number of reports is reduced, 
is another method for data gathering~\cite{R12},~\cite{R13}. In many applications, the spatial or
temporal correlation does not exist between data (e.g. status reports~\cite{R1}), and data
aggregation is a more suitable method for data gathering.

\begin{figure}[t]
\centering
\includegraphics[width=7.5cm]{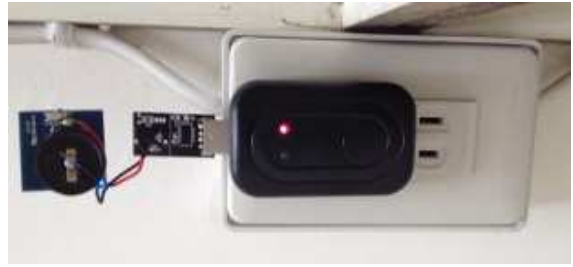}
\caption{A sensor of Octopus wireless sensor network.} \label{fig:OCTOPUS}
\end{figure}

The effectiveness of data aggregation is mainly determined by the routing structure. In many data
aggregation algorithms, a tree is used as the routing 
structure~\cite{R2},~\cite{R5},~\cite{R14},~\cite{R23},~\cite{R24},~\cite{R27}, 
especially for the applications that have to
monitor events continuously. The reason is that sensors, which usually have limited resources, can
save relatively high computational costs for maintaining routing tables if sensors route packets
based on a tree. 

While several papers target at the maximization of the network 
lifetime~\cite{R2},~\cite{R5}, 
the problem of minimizing the total energy cost is also well studied in the literature~\cite{Sinha},~\cite{Min}.
Moreover, for some indoor applications, sensors may have AC power plugs. 
For example, the sensor of Octopus wireless sensor network~\cite{octopus}, 
as shown in Fig.~{\ref{fig:OCTOPUS}}, is plugged in the socket. 
Under such circumstance or energy conservation activity, energy saving then becomes 
the major issue. In this paper, the problem of constructing a data aggregation tree with minimum 
energy cost will be studied. \textbf{Our contributions are described below:}

\begin{itemize}
\item We prove the problem of constructing a data aggregation tree with minimum energy cost,
termed MECAT, is NP-complete and provide a 2-approximation algorithm.
\item We study the variant of such a problem, in which the relay nodes exist, termed MECAT$\_$RN.
We show the MECAT$\_$RN problem is NP-complete and demonstrate a 7-approximation algorithm.
\item We show any $\lambda$-approximation algorithm of the Capacitated Network Design (CND)
problem~\cite{R20} can be used to obtain a $2\lambda$-approximation algorithm of the MECAT$\_$RN problem.
\item We conduct several simulations to evaluate the performances of the proposed algorithms.
\end{itemize}

The remainder of this paper is organized as follows. Section~\ref{Sec: Network Model} describes
the network model and shows the MECAT problem is NP-complete. Section~\ref{Sec: Approximation
Algorithm} provides a 2-approximation algorithm for the MECAT problem. In 
Section~\ref{Sec:
Aggregation with Relay Nodes}, we show the  MECAT$\_$RN problem is NP-complete and give a
7-approximation algorithm. We show a $2\lambda$-approximation algorithm of the MECAT$\_$RN problem
can be obtained using a $\lambda$-approximation algorithm of the CND problem in Section~\ref{Sec:
Discussion}. Using simulations, we evaluate the performances of the proposed algorithms in Section~\ref{Sec: Numerical Results}. Related works are studied in Section~\ref{Sec: Related Work}.
Finally, we conclude the paper in Section~\ref{Sec: Conclusion}.

\section{Network Model and Problem Definition}
\label{Sec: Network Model}

We first illustrate the network model in Section~\ref{Subsec: Network Model}. Subsequently, our
problem is described and shown to be NP-complete in Section~\ref{Subsec: Problem Definition}.

\subsection{The Network Model} \label{Subsec: Network Model}
We model a network as a connected graph $G = (V, E)$ with weights $s(v) \in \mathbb{Z^{+}}$ and 0
associated with each node $v \in V \setminus \{r\}$ and $r$, respectively, where $V$ is the set of
nodes, $E$ is the set of edges, and $r \in V$ is the sink. Each node $v$ has to send a report of
size $s(v)$ to sink $r$ periodically in a multi-hop fashion based on a routing tree. A routing
tree constructed for a network $G = (V, E)$ with sink $r$ is a directed tree $T = (V_T, E_T)$ with
root $r$, where $V_T=V$ and a directed edge $(u,v) \in E_T$ only if an undirected edge $\{u,v\}
\in E$. A node $u$ can send a packet to a node $v$ only if $(u,v) \in E_T$, in which case $u$ is a
child of $v$, and $v$ is the parent of $u$. For the energy consumption, we only consider the
energy cost of the radio~\cite{R5}. Let $Tx$ and $Rx$ be the energy needed to send and receive a
packet, respectively. While routing, a hop-by-hop aggregation is performed according to the
aggregation ratio, $q$, which is the size of reports that can be aggregated into one
packet. Because it would be meaningless if the aggregation
ratio is set to a non-integer, the aggregation ratio is assumed to be an integer through this
paper. It is noteworthy that we implicitly assume that the transmission energy and the receiving 
energy of a packet are constants. In~\cite{Min02topfive, Wang:2001:EEM:383082.383105}, the authors 
observe that in a wireless sensor network with a small packet size, the startup energy cost, that is, 
the energy consumption in the state transition from sleep to idle, exceeds the transmission cost. Thus, we can 
view $Tx$ ($Rx$) as the sum of the transmission cost (receiving cost) and the startup energy cost. 
Then, as long as the packet size is small, $Tx$ and $Rx$ are approximately constants.

%\footnote{Note that the aggregation model is the same as the one using a divisible function
%\cite{R11} as the aggregation ratio is equal to $\vert V \vert$, because a node cannot receive
%more than $\vert V \vert$ reports in a round.}

\begin{figure}
\center
\subfigure[]{\includegraphics[scale=.8]{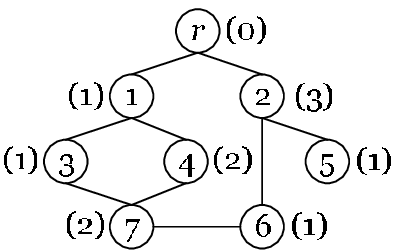}\label{Fig: network model a}}\qquad
\subfigure[]{\includegraphics[scale=.8]{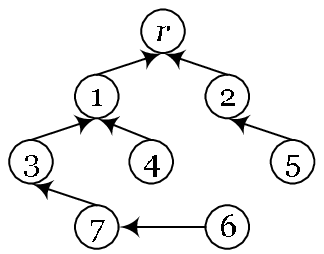}\label{Fig: network model b}}\\
\caption{The network model. (a) A wireless sensor network, where each node has a weight shown in
parentheses. (b) A routing tree.}
\label{Fig: network model}
\end{figure}

\begin{ex}
Fig.~\ref{Fig: network model}(b) is a routing tree constructed for the wireless sensor network
shown in Fig.~\ref{Fig: network model}(a). Assume the aggregation ratio is 3, and both $Tx$ and
$Rx$ are equal to 1. Using the routing tree, node 6 first sends a packet containing its report to
node 7. After node 7 receives the packet from node 6, node 7 aggregates the reports of nodes 6 and
7 into one packet and then sends the packet to node 3. The process proceeds until node $r$
receives the reports of all nodes. Clearly, 3, 2, 2, 1, 1, 1, and 1 packets are sent by nodes 1,
2, 3, 4, 5, 6, and 7, respectively; therefore, a total of 11 packets are sent (and received) by
the nodes. It is easy to verify that a total of 9 packets are required to be sent if the parent of
node 6 is set to node 2.
\end{ex}

\subsection{The Problem and Its Hardness} \label{Subsec: Problem Definition}

We first describe our problem in the following.

\begin{probdefi}
Given a network $G= (V, E)$ with weights $s(v) \in \mathbb{Z^{+}}$ and 0 associated with each node
$v \in V \setminus \{r\}$ and $r$, respectively, a sink $r \in V$, an aggregation ratio $q \in
\mathbb{Z^{+}}$, energy costs $Tx \in \mathbb{R^{+}}$ and $Rx \in \mathbb{R^{+}}$ for transmitting and
receiving a packet, respectively, and $C \in \mathbb{R^{+}}$, the \textbf{M}inimum
\textbf{E}nergy-\textbf{C}ost \textbf{A}ggregation \textbf{T}ree (\textbf{MECAT}) problem asks for
a routing tree $T = (V_T, E_T)$ with root $r$ and $V_T=V$, such that the total transmission and
reception energy consumed by all sensors is not greater than $C$. In addition, MECAT$(G, r, q, Tx,
Rx, C)$ denotes an instance of the MECAT problem, and $COST(T)$ denotes the energy cost of a
routing tree $T$.
\end{probdefi}

Next, we prove that the MECAT problem is NP-complete by showing a polynomial-time reduction from
the Load-Balanced Semi-Matching problem, an NP-complete problem, as described below.

\begin{defi}
A \textbf{semi-matching} in a bipartite graph $G = (U \cup V, E)$ is an edge set $M \subseteq E$,
such that every node in $U$ incident to exactly one edge in $M$. Given a semi-matching $M$ and $v
\in V$, ${Adj}_{M}(v)$ denotes the set of nodes $u$ with $\{v,u\} \in M$.
\end{defi}

\begin{probdefi}~\cite{R6}
Given a bipartite graph $G = (U \cup V, E)$ with a weight $w(u) \in \mathbb{Z^{+}}$ associated with
each node $u \in U$ and $k \in \mathbb{Z^{+}}$, the \textbf{L}oad-\textbf{B}alanced
\textbf{S}emi-\textbf{M}atching (\textbf{LBSM}) problem asks for a semi-matching $M$ such that $k
\ge \max_{v \in V} \sum_{u \in Adj_{M}(v)} w(u)$. Furthermore, LBSM$(G, k)$ denotes an instance of
the LBSM problem.
\end{probdefi}

\begin{thrm}
The MECAT problem is NP-complete. \label{Thm: MECAT NP-completeness}
\end{thrm}

The high-level idea of the proof is to show that finding an aggregation 
tree such that every node sends only one packet, which is a special case of the MECAT problem, is NP-complete. 

\begin{proof}
First, the MECAT problem is clearly in NP, since we can verify in polynomial time if a candidate
solution is a tree and satisfies the energy cost constraint. Next, we prove that the MECAT problem
is NP-hard by showing a polynomial-time reduction from the LBSM problem to the MECAT problem. For
every instance LBSM$(G=(U \cup V,E),k)$, we construct an instance MECAT$(G', r, q, Tx, Rx, C)$ as
follows:

\begin{enumerate}
\item $G' = (\lbrace r \rbrace \cup U \cup V \cup W, E \cup E_{R} \cup E_{W})$ with weights 1
and 0 associated with nodes in $U \cup V \cup W$ and $\{r\}$, respectively,
\item $q = k+1$,
\item $Tx = Rx = 1$, and
\item $C = 2 (\vert W \vert + \sum_{1 \leqslant i \leqslant \vert U \vert}
{\lceil \dfrac{\vert W_{i} \vert + 1}{q}\rceil + \vert V \vert})$,
\end{enumerate}

\noindent where $W = \bigcup_{1 \leqslant i \leqslant \vert U \vert} W_{i}$, $W_{i} = \lbrace
w_{i,j} \vert 1 \leqslant j \leqslant w(u_{i}) - 1 \rbrace$, $E_{R} = \lbrace \{r, v_i\} \vert 1
\leqslant i \leqslant \vert V \vert \rbrace$, and $E_{W} = \lbrace \{u_i, w_{i,j}\} \vert 1
\leqslant i \leqslant \vert U \vert, 1 \leqslant j \leqslant w(u_{i}) - 1 \rbrace$. Clearly, this
instance is constructed in polynomial time. See Fig.~\ref{Fig: reduction1}, for example.

\begin{figure}
\center \subfigure[]{\includegraphics[width = 7.5cm]{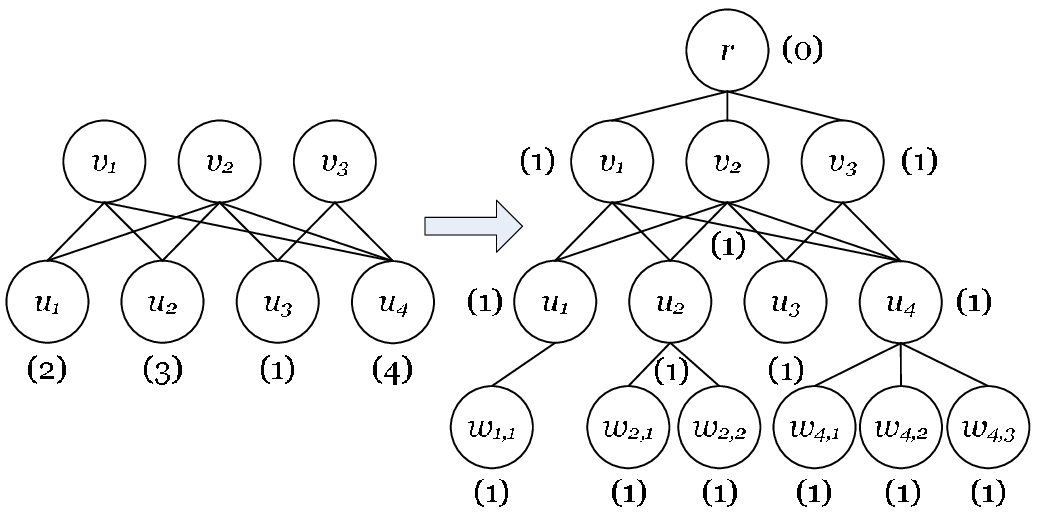}\label{Fig: reduction1 a}}\qquad
\subfigure[]{\includegraphics[width = 7.5cm]{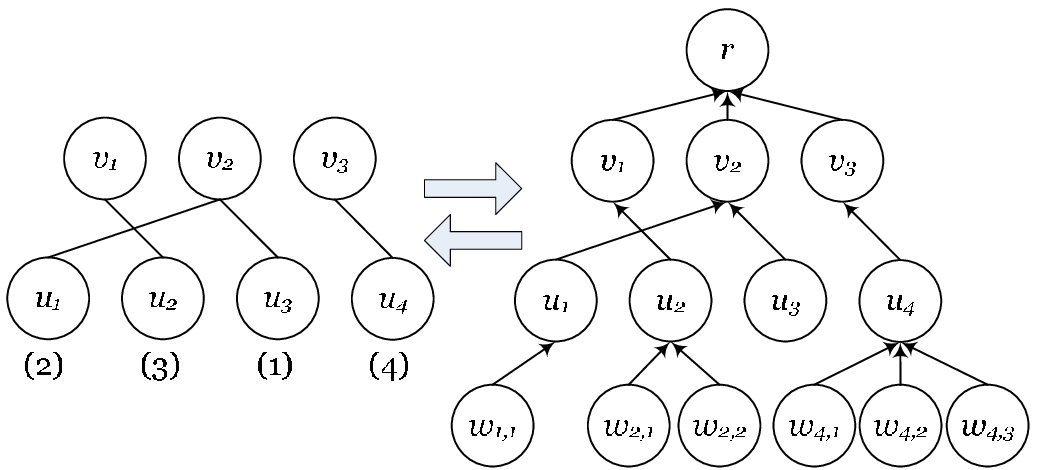}\label{Fig: reduction1 b}}\\
\caption{NP-hardness of the MECAT problem. (a) Reduction from the LBSM problem to the MECAT
problem, where $w(u_i)$ is shown in parentheses and $k$ is equal to 4. (b) Correspondence between
the solutions for the LBSM and MECAT problems.} \label{Fig: reduction1}
\end{figure}

We need to show LBSM$(G,k)$ has a feasible solution if, and only if, MECAT$(G', r, q, Tx, Rx, C)$
has a feasible solution. For the "only if" part, let $M$ be a semi-matching in $G$ such that $k \ge
\max_{v \in V} \sum_{u \in Adj_{M}(v)} w(u)$. We show that using $M$, a routing tree that spans
$\lbrace r \rbrace \cup U \cup V \cup W$ and has a total energy cost not greater than $C$ can be
constructed. Let $T = (V_T, E_T)$ be a routing tree with $V_T=\lbrace r \rbrace \cup U \cup V \cup
W$ and $E_T=\lbrace (v_i,r) \vert 1 \leqslant i \leqslant \vert V \vert \rbrace \cup \lbrace
(u_{i}, v_{j}) \vert \{u_{i}, v_{j}\} \in M \rbrace \cup \lbrace (w_{i,j},u_i) \vert 1 \leqslant i
\leqslant \vert U \vert, 1 \leqslant j \leqslant w(u_{i}) - 1 \rbrace$. We show the total energy
cost of $T$ is not greater than $C$. Since $Tx = Rx = 1$, the total transmission (or reception)
energy of the nodes in $W$ (or $U$) is $\vert W \vert$ and the total transmission (or reception)
energy of the nodes in $U$ (or $V$) is $\sum_{1 \leqslant i \leqslant \vert U \vert}{\lceil
\frac{\vert W_{i} \vert + 1}{q} \rceil}$. Since $C=2 (\vert W \vert + \sum_{1 \leqslant i
\leqslant \vert U \vert}{\lceil \dfrac{\vert W_{i} \vert + 1}{q}\rceil} + \vert V \vert)$ and
$Tx=Rx=1$, we only need to show each node in $V$ sends exactly one packet. The size of reports
sent by node $v_{j}$ in $V$ is
\begin{align}
\nonumber 1+\sum_{u_{i} \in Adj_{M}(v_j)}{(1 + \vert W_{i}\vert)}
& \leqslant 1+\max_{v \in V}{\sum_{u_{i} \in Adj_{M}(v)}{(1 + \vert W_{i}\vert)}} \\
\nonumber & = 1+\max_{v \in V}{\sum_{u_{i} \in Adj_{M}(v)}{w(u_i)}}\\
& \le k+1.
\end{align}
Thus, each node in $V$ needs to send reports with a total size at most $k+1$, which can be
aggregated into 1 packet.

For the "if" part, let $T = (V_T, E_T)$ be a routing tree with minimum energy cost not greater than
$C$. We show a semi-matching in $G$ such that $k \ge \max_{v \in V} \sum_{u \in Adj_{M}(v)} w(u)$
can be constructed using $T$. Clearly, $(v_j,u_i) \notin E_T$ for all $v_j \in V$ and $u_i \in U$;
otherwise, there exists a routing tree $T' = (V_T, E_T \setminus \{(v_j,u_i)\} \cup \{(v_j,r)\})$
with less energy cost than $T$. Let $M=\lbrace \{u_{i}, v_{j}\} \vert (u_{i}, v_{j}) \in E_T
\rbrace$. We show
\begin{equation}\label{Eq: AA3}
k \ge \max_{v \in V} \sum_{u \in Adj_{M}(v)} w(u).
\end{equation}
As in the proof of the "only if" part, the total transmission and reception energy of the nodes in $U$
and $W$ plus the reception energy of the nodes in $V$ is
\begin{equation}\label{Eq: AA1}
2 (\vert W \vert + \sum_{1 \leqslant i \leqslant \vert U \vert}{\lceil \dfrac{\vert W_{i} \vert +
1}{q} \rceil}).
\end{equation}
In addition, the total energy cost of $T$ is not greater than
\begin{equation}\label{Eq: AA2}
C= 2 (\vert W \vert + \sum_{1 \leqslant i \leqslant \vert U \vert}{\lceil \dfrac{\vert W_{i} \vert
+ 1}{q}\rceil + \vert V \vert}).
\end{equation}
(\ref{Eq: AA1}) and (\ref{Eq: AA2}) imply each node in $V$ sends only one packet. Thus, each node
in $V$ receives at most $k$ reports, implying (\ref{Eq: AA3}).
\end{proof}

\section{Approximation Algorithm} \label{Sec: Approximation Algorithm}
As the MECAT problem is NP-complete, we provide an approximation algorithm. Observe that while
sending a packet to the sink, the longer the routing path is, the greater the energy cost is.
Naturally, we would route each packet via a shortest path to the sink. The resulting routing
structure is then a shortest path tree. There are at least three benefits to route packets using a
shortest path tree. First, a shortest path tree is easy to construct in a distributed manner, as
described in the following two steps. The sink node first broadcasts a message such that each node
can evaluate the hop distance from the sink~\cite{R35}. Then, each node sets its parent to the
node with a smaller hop distance from the sink. Second, in many time-critical applications, it is
necessary to route packets using a shortest path tree to achieve the minimum packet transmission
delay. Third, the algorithm is irrelevant to the aggregation ratio and 
report sizes. So it can be applied to the case where the aggregation ratio or report sizes are not 
known or even they may vary from time to time. Although a shortest path tree may not have minimum energy cost (see Fig.~\ref{Fig: shortest
path tree without minimum energy cost}, for example), Theorem~\ref{Thm: Approximation Algorithm}
shows a shortest path tree algorithm has an approximation ratio of 2. 
Definition~\ref{Defi: PP} and
Lemma~\ref{Lemma: PP} are necessary for the proof of Theorem~\ref{Thm: Approximation Algorithm}.

\begin{defi}\label{Defi: PP}
Given a graph $G = (V, E)$ and a root $r \in V$, a \textbf{Minimum Descendant Tree} is a tree $T$
rooted at $r$ and spanning $V$, such that $\sum_{v \neq r}{des_{T}(v)}$ is minimized, where
$des_{T}(v)$ is the total size of reports to be sent by $v$'s descendants in $T$.
\end{defi}

\begin{figure}
\centering \subfigure[]{\includegraphics[scale=.8]{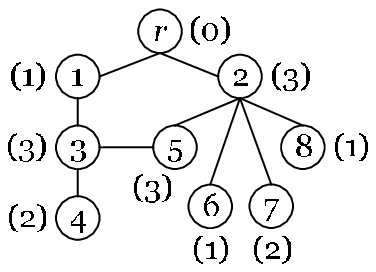}}\qquad
\subfigure[]{\includegraphics[scale=.8]{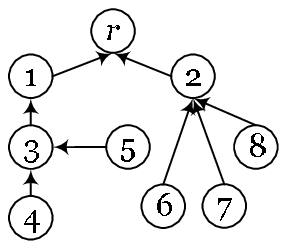}}\\
\caption{(a) A wireless sensor network, where $q = 9$ and $Tx = Rx = 1$. (b) A non-shortest path
tree with minimum energy cost.} \label{Fig: shortest path tree without minimum energy cost}
\end{figure}

\begin{lemma}\label{Lemma: PP}
Every shortest path tree is a minimum descendant tree.
\end{lemma}

\begin{proof}
The lemma directly follows two claims below:

\begin{enumerate}
\item Every minimum descendant tree is a shortest path tree.
\item Every shortest path tree $T$ has the same value of $\sum_{v \neq r}{des_T(v)}$.
\end{enumerate}

We show claim 1 by contradiction. Suppose that there exists a minimum descendant tree
$T'=(V_{T'},E_{T'})$ that is not a shortest path tree. Let $D_{G}(v)$ and $D_{T'}(v)$ be the hop
distances from $v$ to $r$ in $G$ and $T'$, respectively. Let $V' = \lbrace v \vert D_{G}(v) <
D_{T'}(v) \rbrace$ and $v' = \argmin_{v \in V'}{D_{G}(v)}$, i.e., $v'$ is the node in $V'$ with
minimum hop distance to $r$ in $G$. Then, $V' \neq \emptyset$ and $v'$ must exist. Let $u$ be the
parent of $v'$ in $T'$, and $u'$ be $v'$ neighboring node with a smaller hop distance from $r$ in
$G$. Let $T"=(V_{T'},E_{T'} \setminus \{(v',u)\} \cup \{(v',u')\})$. Clearly, $T"$ is a tree with
$\sum_{v \neq r}{des_{T"}(v)} < \sum_{v \neq r}{des_{T'}(v)}$, a contradiction.

For claim 2, let $T_1$ and $T_2$ be any two shortest path trees. Clearly, $T_1$ and $T_2$ have the
same height, say $H$. Let $L_{k}(T)$ be the set of nodes whose hop distances from the root in a
tree $T$ are $k$. We have
\begin{equation}\label{1}
\sum_{v \in L_{H}(T_1)}{des_{T_1}(v)} = \sum_{v \in L_{H}(T_2)}{des_{T_2}(v)} = 0,
\end{equation}
and
\begin{align}\label{2}
\nonumber \sum_{v \in L_k(T)}{des_{T}(v)} = \sum_{v \in L_{k+1}(T)}{des_{T}(v)} &+ \sum_{v \in
L_{k+1}(T)}{s(v)},\\
&\forall 1 \leqslant k \leqslant H-1.
\end{align}
Since $T_1$ and $T_2$ are shortest path trees,
\begin{equation}\label{3}
\sum_{v \in L_{k}(T_1)}{s(v)} = \sum_{v \in L_{k}(T_2)}{s(v)}, \forall 1 \leqslant k \leqslant H.
\end{equation}
By \eqref{1}, \eqref{2}, and \eqref{3}, we have claim 2.
\end{proof}

\begin{thrm}
Every shortest path tree algorithm is a 2-approximation algorithm. \label{Thm: Approximation Algorithm}
\end{thrm}

\begin{proof}
Let $T$ be a routing tree. Since the number of packets sent by nodes equals that received by nodes
in $T$,
\begin{equation}\label{5}
COST(T) = (Tx + Rx)\sum_{v \neq r}{\lceil \frac{des_{T}(v)+s(v)}{q}\rceil}.
\end{equation}
Let $T_{OPT}$ be a routing tree with minimum energy cost and $T_{SPT}$ be an arbitrary shortest
path tree. By Lemma~\ref{Lemma: PP}, we obtain
\begin{align}\label{7}
\nonumber COST(T_{OPT}) &= (Tx + Rx) \sum_{v \neq r}{\lceil \frac{des_{T_{OPT}}(v)+s(v)}{q} \rceil} \\
\nonumber               &\geqslant (Tx + Rx) \sum_{v \neq r}{\frac{des_{T_{OPT}}(v)+s(v)}{q}}\\
                        &\geqslant (Tx + Rx) \sum_{v \neq r}{\frac{des_{T_{SPT}}(v)+s(v)}{q}}.
\end{align}
In addition, by \eqref{5}, we have
\begin{equation}\label{8}
COST(T_{SPT}) = (Tx + Rx) \sum_{v \neq r}{\lceil \frac{des_{T_{SPT}}(v)+s(v)}{q} \rceil} \\
\end{equation}
Therefore, by \eqref{7} and \eqref{8}, we get
\begin{equation}\label{9}
COST(T_{SPT})-COST(T_{OPT}) < (Tx+Rx)(\vert V \vert -1).
\end{equation}
In addition, each node has to send at least one packet, and these packets must be received
by some nodes. Thus,
\begin{equation}\label{10}
COST(T_{OPT}) \geqslant (Tx+Rx)(\vert V \vert -1).
\end{equation}
Combining \eqref{9} and \eqref{10}, we obtain
\begin{equation}
COST(T_{SPT}) < 2 \cdot COST(T_{OPT}).
\end{equation}
\end{proof}

\section{Data Aggregation with Relay Nodes}
\label{Sec: Aggregation with Relay Nodes}

To improve the network connectivity or survivability, the relay node placement problem in a wireless sensor
network has been extensively investigated in the 
literature~\cite{R16},~\cite{R17},~\cite{R18}. 
These relay nodes, which do not produce reports, are used to forward the
packets received from other nodes. In this section, we study the problem of constructing a data
aggregation tree with minimum energy cost in the presence of relay nodes.

\subsection{The Problem and Its Hardness} \label{Subsec: Problem Definition 2}
Here, a routing tree only needs to span all non-relay nodes. For the convenience of
description, we assume every relay node has a zero-sized report. In the following, the problem
is described and shown to be NP-complete. 

\begin{probdefi}
Given a network $G= (V, E)$ with weights $s(u) \in \mathbb{Z^{+}}$ and 0 associated with each
source $u \in U \subseteq V \setminus \{r\}$ and $v \in V \setminus U$, respectively, a set of
sources $U$, a sink $r \in V$, an aggregation ratio $q \in \mathbb{Z^{+}}$, energy costs $Tx \in
\mathbb{R^{+}}$ and $Rx \in \mathbb{R^{+}}$ for transmitting and receiving a packet, respectively, and
$C \in \mathbb{R^{+}}$, the \textbf{M}inimum \textbf{E}nergy-\textbf{C}ost \textbf{A}ggregation
\textbf{T}ree with \textbf{R}elay \textbf{N}odes (\textbf{MECAT$\_$RN}) problem asks for a routing
tree $T = (V_T, E_T)$ with root $r$ and $V_T \supseteq U\cup\{r\}$, such that the total transmission and
reception energy consumed by all sensors is not greater than $C$. Moreover, MECAT$\_$RN$(G, U, r,
q, Tx, Rx, C)$ denotes an instance of the MECAT$\_$RN problem, and $COST(T)$ denotes the energy
cost of a routing tree $T$.
\end{probdefi}

It is natural to guess the MECAT$\_$RN problem is NP-complete, due to its similarity to Problem 1. 
In fact, when $q = \vert U \vert$ and $s(u) = 1$ for all $u \in U$, the MECAT$\_$RN problem becomes
Steiner tree problem~\cite{R7}. On the other hand, for Problem 1, we can add a relay node $w$ and replace any
edge $\{u,v\}$ by two edges $\{u,w\}$ and $\{w, v\}$ to construct an instance of the MECAT$\_$RN problem.
However, for most cases, $q$ is smaller than $\vert U \vert$. Furthermore, when relay nodes 
are deployed, some relay nodes might have degree more than two, that is, some relay nodes can forward packets 
from two or more nodes. By the following theorem, we prove that even if $q < \vert U \vert$ and some relay nodes 
have degree more than two in the original graph, the problem is still NP-complete.

\begin{figure}
\center
\subfigure[]{\includegraphics[scale=.8]{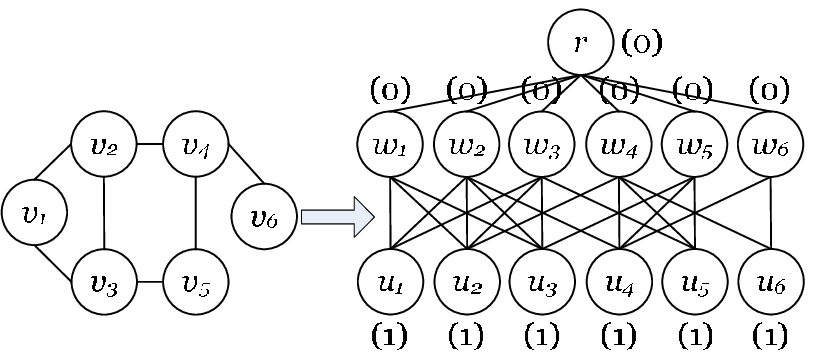}\label{Fig: reduction2 a}}\qquad
\subfigure[]{\includegraphics[scale=.8]{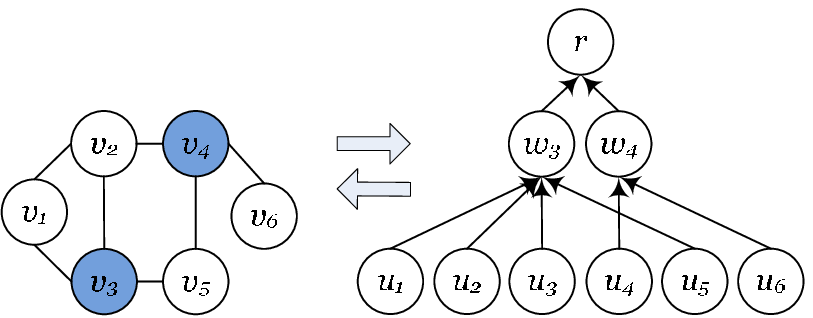}\label{Fig: reduction2 b}}\\
\caption{NP-hardness of the MECAT$\_$RN problem. (a) Reduction from the Dominating Set problem to
the MECAT$\_$RN problem. (b) Correspondence between the solutions for the Dominating Set and
MECAT$\_$RN problems.} \label{Fig: reduction 2}
\end{figure}

\begin{thrm}
The MECAT$\_$RN problem is NP-complete even if $q < \vert U \vert$ and some relay nodes 
have degree more than two in the original graph. \label{Thm: MECATRN NP-completeness}
\end{thrm}

\begin{proof}
First, it is easy to see that the problem is in NP since a non-deterministic algorithm just needs
to guess a tree spanning all nodes in $U$ and check in polynomial time if the energy
cost of the tree is not greater than $C$. Next, to show the MECAT$\_$RN problem is NP-hard, we demonstrate a
polynomial-time reduction from the Dominating Set problem~\cite{R7}, which asks for a dominating
set $D$ in $G$ with $\vert D \vert \leqslant k$ for a given instance DS$(G, k)$, to the MECAT$\_$RN
problem. For any instance DS$(G = (V, E), k)$, we construct an instance MECAT$\_$RN$(G', U, r, q,
Tx, Rx, C)$ as follows:

\begin{enumerate}
\item $G' = (\lbrace r \rbrace \cup W \cup U, E_S \cup E_R)$ with weights 0 and 1 associated with
nodes in $\lbrace r \rbrace \cup W$ and $U$, respectively,
\item $q = \Delta(G)+1$,
\item $Tx = Rx = 1$, and
\item $C = 2(\vert V \vert + k)$,
\end{enumerate}

\noindent where $W = \lbrace w_i \vert 1 \leqslant i \leqslant \vert V \vert \rbrace$, $U =
\lbrace u_i \vert 1 \leqslant i \leqslant \vert V \vert \rbrace$, $E_S = \lbrace \{r,w_i\} \vert 1
\leqslant i \leqslant \vert V \vert \rbrace$, and $E_R = \lbrace \{w_i, u_i\} \vert 1 \leqslant i
\leqslant \vert V \vert \rbrace \cup \lbrace \{w_i, u_j\} \vert \{v_i, v_j\} \in E \rbrace$, and
$\Delta(G)$ denotes the maximum degree of $G$. Clearly, this instance is constructed in polynomial
time and $q < \vert U \vert$. Note that the dominating set problem is trivial when 
$\Delta(G)+1 = \vert U \vert$. In addition, there must exist a relay node with degree more than two.
See Fig.~\ref{Fig: reduction 2}, for example.

We need to show DS$(G, k)$ has a feasible solution if, and only if, MECAT$\_$RN$(G', U, r, q, Tx,
Rx, C)$ has a feasible solution. For the "only if" part, let $D$ be a dominating set in $G$ with
$\vert D \vert \leqslant k$. We show that a routing tree that spans $U$ and has a total energy
cost not greater than $C$ can be constructed using $D$. Let $W' = \lbrace w_i \vert v_i \in D
\rbrace$. We construct $T= (V_T, E_T)$ as follows. Let $V_T=\lbrace r \rbrace \cup W' \cup U$.
Set the parent of $w_i$ to $r$ for all $w_i \in W'$ and the parent of $u_i$ to an
arbitrary neighboring node in $W'$ for all $u_i \in U$. Since $D$ is a dominating set in $G$, each
node in $U$ has a parent in $T$. Thus, $T$ is a routing tree spanning $U$. We show the total
energy cost of $T$ is not greater than $C$. Since $Tx = Rx = 1$, the total transmission (or
reception) energy of the nodes in $U$ (or $W$) is $\vert U \vert = \vert V \vert$. In addition,
since each node in $W'$ receives at most $\Delta(G)+1$ reports, it sends exactly one packet to the
parent in $T$. Since $\vert W' \vert = \vert D \vert \le k$, the total transmission (or
reception) energy of the nodes in $W$ (or $r$) is at most $k$. Thus,  the total energy cost of
$T$ is not greater than $2 \cdot (\vert V \vert+k)=C$.

For the "if" part, let $T= (V_T, E_T)$ be a routing tree that spans $U$ and has minimum energy cost
not greater than $C= 2(\vert V \vert + k)$. Let $W'= V_T \setminus (U \cup \{r\})$. We claim that $D
= \lbrace v_i \vert w_i \in W' \rbrace$ is a dominating set in $G$ with $\vert D \vert \leqslant k$.
In $T$, the parent of each node in $U$ is in $W'$. This implies $D$ is a dominating set
in $G$. In addition, as in the proof of Theorem~\ref{Thm: MECAT NP-completeness}, the parent of
each node in $W'$ is $r$; otherwise, a routing tree that spans $U$ and has less energy cost exists.
Since the total transmission (or reception) energy of the nodes in $U$ (or $W'$) is $\vert V
\vert$, the total transmission (or reception) energy of the nodes in $W'$ (or $r$) is at most
$k$. Thus, $\vert D \vert = \vert W' \vert \leqslant k$.
\end{proof}

\subsection{Approximation Algorithm} \label{Subsec: Approximation of MECATRN}
A Steiner tree algorithm and a shortest path tree algorithm provide solutions with minimum number
of edges and minimum average hop distance from sources to the sink for the MECAT$\_$RN problem,
respectively. However, both of them have bad approximation ratios, as described in Theorems~\ref{Thm: Steiner tree algorithm approximation ratio} and \ref{Thm: shortest path tree algorithm
approximation ratio}. Their proofs are given in the appendix.

\begin{thrm} \label{Thm: Steiner tree algorithm approximation ratio}
The approximation ratio of a Steiner tree algorithm is at least $\Theta(\vert U \vert)$.
\end{thrm}

\begin{thrm} \label{Thm: shortest path tree algorithm approximation ratio}
The approximation ratio of a shortest path tree algorithm is at least $\Theta(\vert U \vert)$.
\end{thrm}

Theorems~\ref{Thm: Steiner tree algorithm approximation ratio} and \ref{Thm: shortest path tree
algorithm approximation ratio} tell us that a routing tree with a constant approximation ratio
cannot be found by minimizing either the number of edges or the average hop distance from sources
to the sink. Our method (Algorithm~\ref{Algo: Our Algo}) is to construct a routing tree that
approximates both a Steiner tree and a shortest path tree based on 
Salman's algorithm~\cite{R21}
(Algorithm~\ref{Algo: Salman Algo}) for the Capacitated Network Design 
problem~\cite{R20}. The
Capacitated Network Design problem, Salman's algorithm, and the Light Approximate Shortest-path
Tree (LAST)~\cite{R22} used in Salman's algorithm are introduced below.

\begin{probdefi}~\cite{R20}
Given a graph $G=(V, E)$ with weight $w(e) \in \mathbb{R^{+}}$ associated with each edge $e \in E$
indicating the length and weight $s(u)  \in \mathbb{Z^{+}}$ associated with each source $u \in U
\subseteq V$ indicating the demand size to route to sink $r \in V$, a set of
sources $U$, a sink $r$, and a transmission facility capacity $q \in \mathbb{Z^{+}}$, the
\textbf{C}apacitated \textbf{N}etwork \textbf{D}esign (\textbf{CND}) problem is to find a path
from $u$ to sink $r$ for each source $u \in U$, such that the total cost of installing all
facilities is minimized, where the cost of installing $k$ facilities on an edge with length $l$ is
$k \cdot l$. Note that a node might have multiple outgoing edges in a feasible solution of the CND
problem. That is, a feasible solution of the CND problem might not be a tree. Moreover, CND$(G, U,
r, q, C)$ denotes an instance of the CND problem, and $COST_{CND}(R)$ denotes the cost of
installing facilities of a route $R$.
\end{probdefi}

\begin{defi}~\cite{R22}
Given a graph $G=(V,E)$ with weight $w(e) \in \mathbb{R^{+}}$ associated with each edge $e \in E$,
a spanning tree $T$ rooted at $r$ is called an \textbf{$(\alpha, \beta)$-LAST}, where $\alpha \geqslant
1$ and $\beta \geqslant 1$, if the following two conditions are satisfied:
\begin{enumerate}
\item For every node $v$, the distance from $v$ to $r$ in $T$ is at most $\alpha$ times the
minimum distance from $v$ to $r$ in $G$.
\item The weight of $T$ is at most $\beta$ times that of the minimum spanning tree of $G$.
\end{enumerate}
\end{defi}

\begin{algorithm}[t]
\caption{: Salman's Algorithm for the CND Problem} \label{Algo: Salman Algo}
\renewcommand{\algorithmicrequire}{\textbf{Input:}}
\begin{algorithmic}[1]
\REQUIRE $G$, $U$, $r$, $C$
\STATE Construct a complete graph $G'$ with node set $U \cup \lbrace r \rbrace$.\\
\STATE Set the weight of each edge $(u,v)$ in $G'$ to the length of the shortest path from $u$ to
$v$ in $G$.\\
\STATE Compute a (3,2)-LAST $T_L$ in $G'$.\\
\STATE Let $(u,u_1,\cdots,u_n,r)$ be the shortest path from $u$ to $r$ in $T_{L}$. Then, the
concatenation of paths $P_{u,u_1}$, $P_{u_1,u_2}$, $\cdots$, and $P_{u_n,r}$ is the output path
from $u$ to $r$, where $P_{x,y}$ denotes the shortest path from $x$ to $y$ in $G$.\\
\STATE Return the output path from $u$ to $r$ for each $u \in U$.
\end{algorithmic}
\end{algorithm}

\begin{algorithm}[t]
\caption{: Our Algorithm for the MECAT$\_$RN Problem}\label{Algo: Our Algo}
\renewcommand{\algorithmicrequire}{\textbf{Input:}}
\begin{algorithmic}[1]
\REQUIRE $G$, $U$, $r$, $Tx$, $Rx$, $C$
\STATE Construct a complete graph $G'$ with node set $U \cup \lbrace r \rbrace$.\\
\STATE Set the weight of each edge $(u,v)$ in $G'$ to the hop distance from $u$ to $v$ in $G$.\\
\STATE Compute a (3,2)-LAST $T_L$ in $G'$.\\
\STATE Compute $G" =(V",E")$, where $V"=\lbrace w \vert w \in P_{u,v}$ for some $(u,v) \in T_L
\rbrace$, $E"=\lbrace \{x,y\} \vert  \{x,y\} \in P_{u,v}$ for some $(u,v) \in T_L \rbrace$, and
$P_{u,v}$ is the shortest path from $u$ to $v$ in $G$.\\
\STATE Construct a shortest path tree $T_{SPT}$ rooted at $r$ and spanning $U$ in $G"$.\\
\STATE Return $T_{SPT}$.
\end{algorithmic}
\end{algorithm}

Theorem~\ref{Thm: MECATRN Approximation} shows Algorithm~\ref{Algo: Our Algo} is a 7-approximation
algorithm of the MECAT$\_$RN problem. Lemma~\ref{Lemma: Salman}, derived from the proof of Lemma
2.1 in~\cite{R21}, is used in the proof of Theorem~\ref{Thm: MECATRN Approximation}. We omit the
proof of Lemma~\ref{Lemma: Salman} due to the similarity. Also note that 
Algorithm~\ref{Algo: Our Algo} is again irrelevant to the aggregation ratio and report sizes.

\begin{lemma} \label{Lemma: Salman}
Let $R = \bigcup_{u \in U} P_{u,r}$ be a route of the CND problem, where $P_{u, r}$ is the routing
path from source $u$ to sink $r$, and let $R_{OPT}$ be the route with minimum cost of the CND
problem. Then, $COST_{CND}(R) \leqslant (\alpha' + \beta') COST_{CND}(R_{OPT})$, if the following
two conditions are satisfied:
\begin{enumerate}
\item For every source $u$, the length of $P_{u, r}$ is at most $\alpha'$ times the minimum distance
from $u$ to $r$ in $G$.
\item The total lengths of the edges of $R$ is at most $\beta'$ times that of the Steiner tree of $G$
spanning $U$.
\end{enumerate}
\end{lemma}

\begin{thrm}
Algorithm~\ref{Algo: Our Algo} is a 7-approximation algorithm of the MECAT$\_$RN problem.
\label{Thm: MECATRN Approximation}
\end{thrm}

\begin{proof}
Let Algorithm $\mathcal{A}$ be obtained from Algorithm~\ref{Algo: Our Algo} by replacing Line 2
with Line 2 of Algorithm~\ref{Algo: Salman Algo} and modifying Line 6 to return the path from $u$
to $r$ in $T_{SPT}$ for each $u \in U$ instead of  $T_{SPT}$. We first claim that Algorithm
$\mathcal{A}$ is a 7-approximation of the CND problem. Let $R_1$ and $R_\mathcal{A}$ be the
solutions generated by Algorithms 1 and $\mathcal{A}$, respectively. Clearly, the following two
facts hold:
\begin{enumerate}
\item For every source $u$, the length of $P_{u, r}$ in $R_\mathcal{A}$ is less than that in $R_1$.
\item The total lengths of the edges of $R_\mathcal{A}$ is is less than that of $R_1$.
\end{enumerate}
\noindent \cite{R21} tells us that the length of $P_{u, r}$ in $R_1$ is at most 3 times the minimum
distance from $u$ to $r$ in $G$ for every source $u$ and the total lengths of the edges of $R_1$
is at most 4 times that of the Steiner tree of $G$ spanning $U$. Thus, Algorithm $\mathcal{A}$ is
a 7-approximation of the CND problem by Lemma 2.

Next, given MECAT$\_$RN$(G_1, U, r, q, Tx, Rx, C)$, we construct CND$(G_2, U, r, q, C)$, where
$G_2$ is obtained from $G_1$ by setting the weight of each edge to $Tx+Rx$. Let $T_2$ and
$T_{OPT}$ be a routing tree generated by Algorithm 2 and the routing tree with minimum energy cost
for MECAT$\_$RN$(G_1, U, r, q, Tx, Rx, C)$, respectively. Let $R_\mathcal{A}$ and $R_{OPT}$ be a
route generated by Algorithm $\mathcal{A}$ and the route with minimum cost of installing
facilities for CND$(G_2, U, r, q, C)$, respectively. Note that for each $u \in U$, the sequence of
the nodes in the path from $u$ to $r$ in $T_2$ is equal to that in the path from $u$ to $r$ in
$R_\mathcal{A}$. Thus,
\begin{equation}\label{Eq: RR1}
COST(T_2)=COST_{CND}(R_\mathcal{A}).
\end{equation}
It is also noted that a collection of the path from $u$ to $r$ in $T_{OPT}$ for each $u \in U$ can
be a route $R$ for CND$(G_2, U, r, q, C)$, in which case $COST(T_{OPT})=COST_{CND}(R)$. It implies
\begin{equation}\label{Eq: RR2}
COST_{CND}(R_{OPT}) \le COST(T_{OPT}).
 \end{equation}
Combining (\ref{Eq: RR1}) and (\ref{Eq: RR2}) together with the fact that
$COST_{CND}(R_\mathcal{A}) \le 7 COST_{CND}(R_{OPT})$, we obtain
\begin{equation}\label{Eq: RR4}
COST(T_2) \le 7 COST(T_{OPT}).
\end{equation}
\end{proof}

\section{Discussion} \label{Sec: Discussion}
In Section~\ref{Subsec: Approximation of MECATRN}, we obtain a 7-approximation algorithm of the
MECAT$\_$RN problem from Salman's 7-approximation algorithm of the CND problem. In this section,
we show any $\lambda$-approximation algorithm of the CND Problem $\mathcal{A}$ can be used to
obtain a $2\lambda$-approximation algorithm of the MECAT$\_$RN problem, as described in Algorithm~\ref{Algo: CND-MECATRN} and Theorem~\ref{Thm: CND-MECATRN}. See Fig.~\ref{Fig: alg3} for an example.

\begin{algorithm}
\caption{: CND-Based Algorithm for the MECAT$\_$RN Problem}\label{Algo: CND-MECATRN}
\renewcommand{\algorithmicrequire}{\textbf{Input:}}
\begin{algorithmic}[1]
\REQUIRE $G$, $U$, $r$, $q$, $Tx$, $Rx$, $C$, $\mathcal{A}$
\STATE Obtain a graph $G'$ from $G$ by setting the weight of each edge in $G$ to $Tx+Rx$.\\
\STATE Execute $\mathcal{A}$ with inputs $G'$, $U$, $r$, $q$, and $C$ to obtain $P_{u,r}$, the
path from $u$ to $r$ in $G'$, for each $u \in U$.\\
\STATE Compute $G" =(V",E")$, where $V"=\lbrace w \vert  w \in P_{u,r} \rbrace$ and $E"=\lbrace
\{x,y\} \vert \{x,y\} \in P_{u,r} \rbrace$.\\
\STATE Construct a shortest path tree $T_{SPT}$ rooted at $r$ and spanning $U$ in $G"$.\\
\STATE Return $T_{SPT}$.
\end{algorithmic}
\end{algorithm}

\begin{figure}
\center
\subfigure[]{\includegraphics[scale=.8]{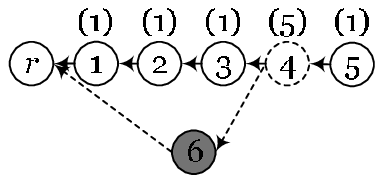}\label{Fig: alg3 a}}\qquad
\subfigure[]{\includegraphics[scale=.8]{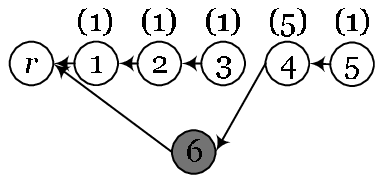}\label{Fig: alg3 b}}\\
\caption{Example of Algorithm 3, where node 6 is a relay node. $q$ is equal to 5. Weights are shown in parentheses. (a) The $P_{u,r}$ after the execution of Line 2, where node 4 sends its report along the path (4, 6, 1) and node 5 sends its report along the path (5, 4, 3, 2, 1). (b) The output of Algorithm 3.}
\label{Fig: alg3}
\end{figure}

\begin{thrm}
Algorithm~\ref{Algo: CND-MECATRN} is a $2\lambda$-approximation algorithm of the MECAT$\_$RN
problem, given that $\mathcal{A}$ is a $\lambda$-approximation algorithm of the CND problem.
\label{Thm: CND-MECATRN}
\end{thrm}

\begin{proof}
Let $P_{u,r}$ be the path from $u$ to $r$ output by algorithm $\mathcal{A}$, and let $R =
\bigcup_{u \in U} P_{u,r}$. We have
\begin{equation}\label{Eq: PP4}
\frac{COST_{CND}(R)}{Tx+Rx} = \sum_{u \in U}{\frac{s(u)}{q}l(P_{u,r})} + \sum_{e \in
R}{(\lceil z(e) \rceil - z(e))},
\end{equation}
where $z(e) = \sum_{u:e \in P_{u,r}}{s(u)/q}$ and $l(p)$ is the length of path $p$. Let $T_{SPT}$
be the routing tree generated by Algorithm~\ref{Algo: CND-MECATRN}, and let $R' = \bigcup_{u \in
U} P'_{u,r}$, where $P'_{u,r}$ denotes the path from $u$ to $r$ in $T_{SPT}$. Then,
\begin{align}\label{Eq: PP1}
\nonumber \frac{COST(T_{TSP})}{Tx+Rx} &= \sum_{u \in U}{\frac{s(u)}{q}l(P'_{u,r})} +
\sum_{e \in R'}{(\lceil z'(e) \rceil - z'(e))}\\
&< \sum_{u \in U}{\frac{s(u)}{q}l(P'_{u,r})} + \sum_{e \in
R'}{1},
\end{align}
where $z'(e) = \sum_{u:e \in P'_{u,r}}{s(u)} / q$. Let $R_{OPT}$ be the route with minimum cost of
installing facilities. (\ref{Eq: PP4}) together with the fact that $COST_{CND}(R) \leqslant
\lambda COST_{CND}(R_{OPT})$ implies
\begin{align}\label{Eq: PP2}
\sum_{u \in U}{\frac{s(u)}{q}l(P'_{u,r})} \leqslant \sum_{u \in
U}{\frac{s(u)}{q}l(P_{u,r})} \leqslant \frac{\lambda COST_{CND}(R_{OPT})}{Tx+Rx}.
\end{align}
In addition,
\begin{equation}\label{Eq: PP3}
\sum_{e \in R'}{1} \leqslant \sum_{e \in R}{1} \leqslant \frac{COST_{CND}(R)}{Tx+Rx}\leqslant
\frac{\lambda COST_{CND}(R_{OPT})}{Tx+Rx}.
\end{equation}
 By (\ref{Eq: PP1}), (\ref{Eq: PP2}), and (\ref{Eq: PP3}), we have
\begin{equation}\label{Eq: PP}
COST(T_{TSP}) < 2 \lambda COST_{CND}(R_{OPT}).
\end{equation}
As in the proof of Theorem~\ref{Thm: MECATRN Approximation},
\begin{equation}\label{Eq: QQ}
COST_{CND}(R_{OPT}) \le COST(T_{OPT}).
\end{equation}
Combining (\ref{Eq: PP}) and (\ref{Eq: QQ}), we obtain
\begin{equation}
COST(T_{TSP}) < 2 \lambda COST(T_{OPT}).
\end{equation}
\end{proof}

When all reports have the same size, Hassin \textit{et al.} propose a $(1 +
\rho_{st})$-approximation algorithm of the CND problem~\cite{R20}, where $\rho_{st}$ denotes the
approximation ratio of the algorithm of the Steiner tree problem. 
Recently, a 1.39-approximation algorithm of the Steiner tree problem 
has been proposed by Byrka \textit{et al.}~\cite{Byrka}.
Therefore ,we can obtain a 4.78-approximation algorithm of the 
MECAT$\_$RN problem by Algorithm~\ref{Algo: CND-MECATRN}. As the reports have different sizes, the algorithm proposed by Hassin \textit{et.
al.} for the CND problem~\cite{R20} has an approximation ratio $(2 +\rho_{st})$, in which case a
6.78-approximation algorithm of the MECAT$\_$RN problem can be obtained by 
Algorithm~\ref{Algo:
CND-MECATRN}. 

\section{Numerical Results} \label{Sec: Numerical Results}
\begin{table}[t]
\caption{Simulation Settings} \centering \label{Tab: Parameter Settings}
\begin{tabular}{|c|c|}
\hline
Number of nodes & 100\\
\hline
Field & $100 \times 100$\\
\hline
Sink location & (50, 50)\\
\hline
$R$ (transmission range) & 20\\
\hline
$q$ (aggregation ratio) & 2, 4, 6, ..., 50\\
\hline
Tx (transmission energy cost per packet) & 2\\
\hline
Rx (reception energy cost per packet) & 1\\
\hline
Uniform report size & 1\\
\hline
Non-uniform report size & 1, 2, 3, 4, 5\\
\hline
Probability of being relay nodes & 0.3\\
\hline
\end{tabular}
\label{table:sim} % is used to refer this table in the text
\end{table}
Two simulations were conducted here. In the first and second simulations, algorithms of the MECAT
problem (data aggregation without relay nodes) and the MECAT$\_$RN problem (data aggregation with
relay nodes) were compared, respectively. We also compared our algorithms with the lower bound of
the minimum energy cost $LB$ evaluated by (\ref{LB}).

\begin{equation}\label{LB}
LB = (Tx+Rx)\cdot \max{\lbrace \sum_{u \in U}{\frac{s(u)}{q}l(u,r)}, \vert E(T_{Steiner})\vert
\rbrace},
\end{equation}

\noindent where $U$ is the set of sources, $q$ is the aggregation ratio, $l(u,r)$ is the hop
distance from $u$ to $r$ in a shortest path tree and $\vert E(T_{Steiner})\vert$ is the number of
edges in a Steiner Tree. $LB$ is evaluated by (\ref{LB}) due to the fact that the corresponding
minimum energy cost of the MECAT problem and the MECAT$\_$RN problem is impossible to be smaller
than each of $(Tx+Rx)\cdot \sum_{u \in U}{\frac{s(u)}{q}l(u,r)}$ and $(Tx+Rx) \cdot \vert
E(T_{Steiner})\vert$. Since a Steiner tree cannot be obtained in polynomial time, we use a
2-approximation algorithm to construct a Steiner tree $T(V_{ST},E_{ST})$~\cite{R34}, and evaluate
$|E(T_{Steiner})|$ by (\ref{Edge number of Steiner tree}).

\begin{equation}\label{Edge number of Steiner tree}
|E(T_{Steiner})| = \max{\lbrace \frac{|E_{ST}|}{2}, |U| \rbrace}.
\end{equation}

In a wireless sensor network, 100 sensor nodes were uniformly deployed in a 100 $\times$ 100
field. A link exists between two sensor nodes with distance less than or equal to the transmission
range $R = 20$. Since the transmission power is about two times the reception 
power~\cite{R33},
$Tx$ and $Rx$ are set to 2 and 1, respectively. If all reports have the same size (uniform report
size), the size is set to 1; otherwise (non-uniform report size), the sizes are randomly set to
range from 1 to 5. The energy cost of each algorithm was evaluated under different aggregation
ratios 2, 4, 6, $\cdots$, 50 (2, 4, 6, $\cdots$, 100 for the non-uniform report size). 
In the second simulation, each node has probabilities $0.7$ and
$0.3$ to be a source and a relay node, respectively. Empirical data were obtained by averaging
data of 30 different networks. Table~\ref{Tab: Parameter Settings} summarizes the simulation
settings.

\begin{figure}[t]
\centering
\includegraphics[width=7.5cm]{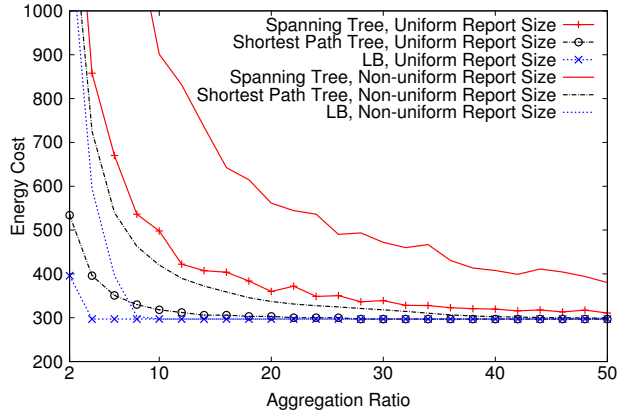}
\caption{Effect of aggregation ratio on energy cost under different algorithms for data
aggregation without relay nodes.} \label{fig:NRUDG}
\end{figure}

\subsection{Results for Data Aggregation without Relay Nodes}
Fig.~\ref{fig:NRUDG} shows the energy cost of algorithms for data aggregation without relay nodes
under different aggregation ratios. It can be seen that the shortest path tree algorithm (the
proposed algorithm for the MECAT problem) significantly outperforms the spanning tree algorithm.
When the aggregation ratio is greater than or equal to the sum of the sizes of the reports sent by
most of the nodes, the energy cost of each algorithm approaches $(Tx +Rx) \cdot \vert V
\vert=300$.

\subsection{Results for Data Aggregation with Relay Nodes}
\begin{figure*}[t]
\center \subfigure[]{\includegraphics[width=8.0cm]{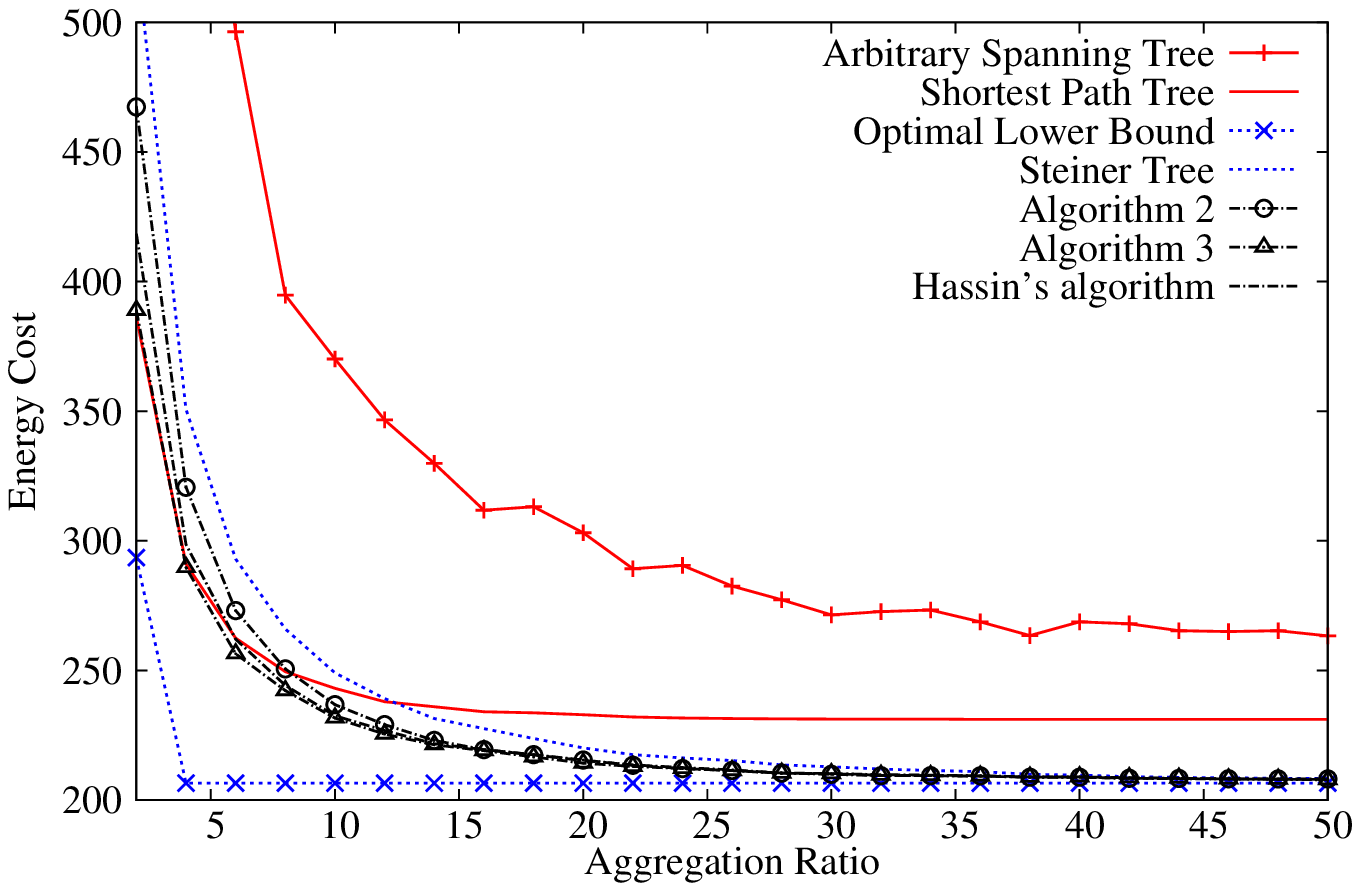}\label{fig:RUUDG1}}\hfil
\subfigure[]{\includegraphics[width=8.0cm]{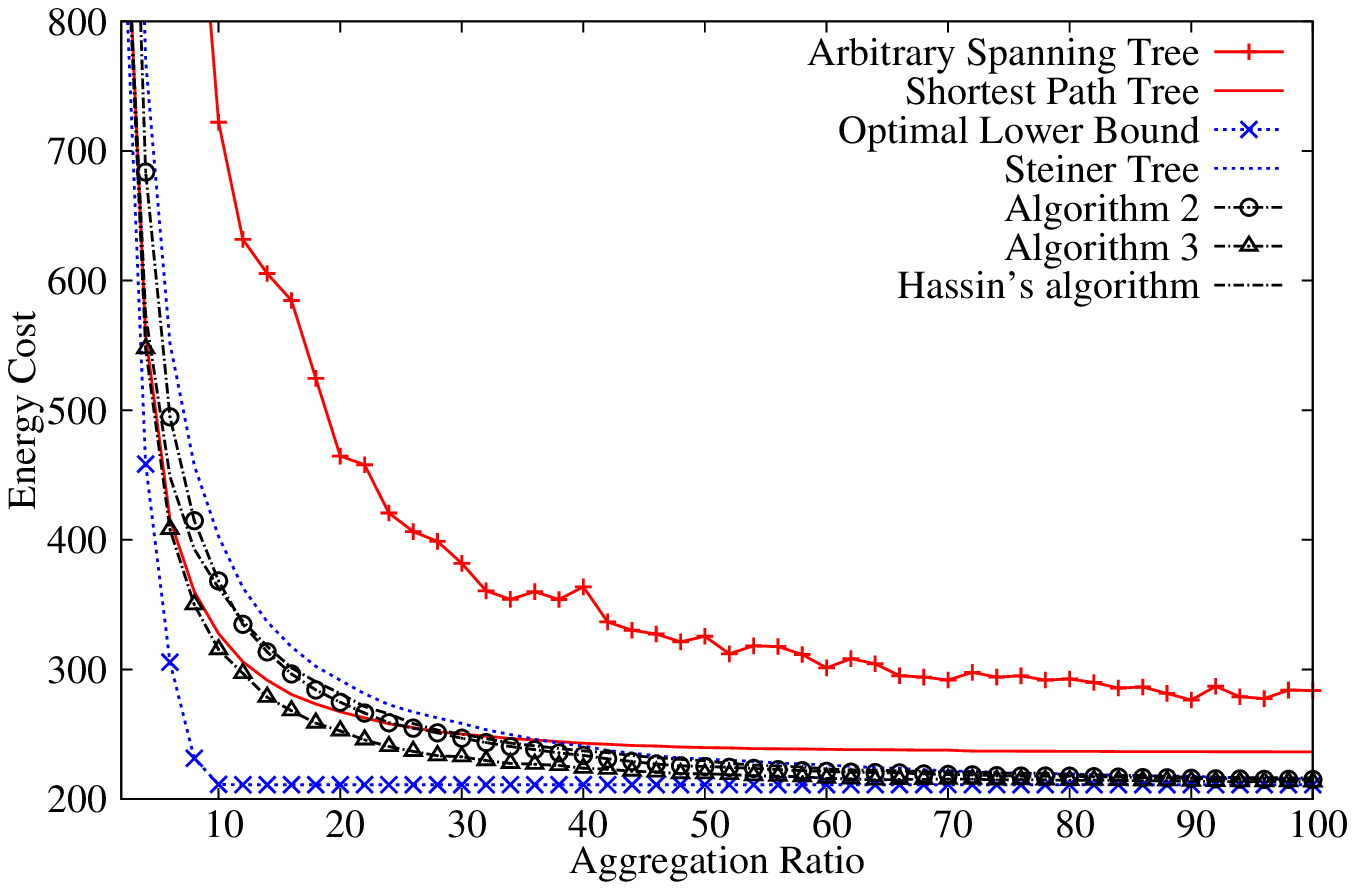}\label{fig:RUUDG2}}\\
\caption{Effect of aggregation ratio on energy cost under different algorithms for data
aggregation with relay nodes. (a) Uniform report size. (b) Non-uniform report size.}
\label{fig:NRUDG2}
\end{figure*}
Fig.~\ref{fig:NRUDG2} shows the energy cost of algorithms for data aggregation with relay nodes
under different aggregation ratios. Algorithms 2 and 3 (the proposed algorithms for the
MECAT$\_$RN problem) construct data aggregation trees based on  Salman's 
algorithm~\cite{R21} and
Hassin's algorithm~\cite{R20}, respectively. A non-tree routing structure established by Hassin's
algorithm~\cite{R20} is also compared here. Four observations are noteworthy. First, the energy
cost of each algorithm approaches $LB$ as the aggregation ratio is great. Second, although a
shortest path tree algorithm and a Steiner tree algorithm have bad performances in the worst cases
(see Theorems~\ref{Thm: Steiner tree algorithm approximation ratio} and \ref{Thm: shortest path
tree algorithm approximation ratio}), they have good performances in average cases. Third, a
shortest path tree performs better and worse than a Steiner tree algorithm when the aggregation
ratio is small and great, respectively. This is because as the aggregation ratio is great, a
packet can carry a large number of reports, and thus, the energy cost highly depends on the number
of edges in the data aggregation tree. On the contrary, the energy cost highly depends on the
lengths of the paths from sources to the sink as the aggregation ratio is small. Fourth, Algorithm
3 outperforms Hassin's algorithm as the reports have different sizes, in which case Hassin's
algorithm usually utilizes only half of the size of a packet, and in contrast, Algorithm 3
utilizes the packet efficiently.

\section{Related Work}\label{Sec: Related Work}
In~\cite{R2}, an algorithm is demonstrated to find the best shortest path tree that maximizes the
network lifetime. In~\cite{R5}, the authors prove the problem of finding an optimal aggregation 
tree that maximizes the network lifetime is NP-complete and propose an approximation algorithm. 
In~\cite{R23}, a randomized $O(1)$-approximation algorithm is given to construct 
a simultaneous optimal aggregation tree based on the geographic correlation of reports. In~\cite{R24}, a data gathering tree is constructed to overcome the changes of the network topology. 
In~\cite{R27}, a data gathering tree is constructed based on the connected dominating set. The goal of 
this paper is to find a distributed data aggregation scheduling such that the time latency is minimized.

In~\cite{R1}, the problem of finding a routing structure minimizing the number of transmitted
packets is studied. They show that routing packets on any two shortest path trees does not significantly
affect the effectiveness of data aggregation. In addition, all reports are assumed to have the
same size and the existence of relay nodes are not taken into consideration, which is different
from this paper.

A problem similar to ours is studied in~\cite{R14}, but the aggregation model and optimization goal are different. 
It is assumed that any $j$ reports can be aggregated into $f(j)$ reports, where $f$ 
is concave. The goal is to minimize the number of transmitted reports. However, in our data aggregation model, 
we assume that any $j$ reports can be aggregated into $\lceil \frac{j}{q} \rceil$ packets. And our 
goal is to minimize the number of transmitted packets.

\section{Conclusion and Future work} \label{Sec: Conclusion}
In this paper, we study the problem of constructing energy-efficient data aggregation trees. Two
types of this problem are investigated: the one without relay nodes and the one with relay nodes.
Both of them are shown to be NP-complete. For the problem without relay nodes, we find that a
shortest path tree algorithm turns out to be a 2-approximation algorithm and can be easily
implemented in a distributed manner. For the problem with relay nodes, we first show that a
shortest path tree algorithm and a Steiner tree algorithm each have bad performance in the worst
cases. We then obtain an $O(1)$-approximation algorithm by constructing a shortest path tree on
the routing structure of the Capacitated Network Design problem. Simulations show that the
proposed algorithms each have good performance in terms of the energy cost. Simulations also show
that for data aggregation with relay nodes, a tree might outperform a non-tree structure in terms
of the energy cost. The reason is in a tree, the data is concentrated in a small number of nodes,
resulting in efficient utilization of packets. 

In this paper, we exploit the feature that the startup energy cost exceeds the transmission energy cost and 
further assume the energy cost $Tx$ and $Rx$ are constants. This assumption is reasonable on low duty-cycle sensor networks, 
since the number of transmitted packets is approximately equal to the number of startups. However, this assumption cannot be 
extend to high duty-cycle sensor networks. Hence, we plan to take the effect of the scheduling into consideration as our 
future work.

\appendix
\subsection{Proof of Theorem~\ref{Thm: Steiner tree algorithm approximation ratio}}
\begin{proof}
Consider one kind of instance of the MECAT$\_$RN problem, MECAT$\_$RN$(G, U, r, q, Tx, Rx, C)$,
where $G = (V, E)$ with weights 1 and 0 associated with each node $u \in U$ and $v \in V \setminus
U$, respectively, $U=\{u_1,u_2,\cdots,u_{\vert U\vert}\}$, $V=\{r\} \cup U \cup \{s_i \vert 3 \le
i \le \vert U\vert\}$, $E=\{\{r,u_1\}\} \cup \{\{u_i,u_{i+1}\} \vert 1 \le i \le \vert U\vert-1\}
\cup \{\{r,s_i\},\{s_i,u_i\} \vert 3 \le i \le \vert U\vert \}$, $q=2$, and $Tx=Rx=1$. See Fig.~\ref{Fig: Steiner tree instance}, for the instance with $\vert U\vert=5$. Clearly, $T_S=(V_S,E_S)$
constructed by a Steiner tree algorithm has energy cost $\Theta(\vert U\vert^2)$, where $V_S=\{r\} \cup
U$ and $E_S=\{(u_1,r)\} \cup \{(u_{i+1},u_i) \vert 1 \le i \le \vert U\vert-1\}$. However,
$T_O=(V_O,E_O)$ has energy cost $\Theta(\vert U\vert)$, where $V_O=V$ and $E_O=\{(u_2,u_1),(u_1,r)\}
\cup \{(u_i,s_i),(s_i,r) \vert 3 \le i \le \vert U\vert \}$.
\end{proof}

\subsection{Proof of Theorem~\ref{Thm: shortest path tree algorithm approximation ratio}}
\begin{proof}
Consider one kind of instance of the MECAT$\_$RN problem, MECAT$\_$RN$(G, U, r, q, Tx, Rx, C)$,
where $G = (V, E)$ with weights 1 and 0 associated with each node $u \in U$ and $v \in V \setminus
U$, respectively, $U=\{u_1,u_2,\cdots,u_{\vert U\vert}\}$, $V=\{r\} \cup U \cup \{s_{i,j} \vert 3
\le i \le \vert U\vert, 1 \le j \le i-2\}$, $E=\{\{r,u_1\}\} \cup \{\{u_i,u_{i+1}\} \vert 1 \le i
\le \vert U\vert-1\} \cup \{\{r,s_{i,1}\},\{s_{i,i-2},u_i\} \vert 3 \le i \le \vert U\vert \} \cup
\{\{s_{i,j},s_{i,j+1}\} \vert 3 \le i \le \vert U\vert, 1 \le j \le i-3\}$, $q=\vert U\vert$, and
$Tx=Rx=1$. See Fig.~\ref{Fig: shortest path tree instance}, for the instance with $\vert
U\vert=5$. Clearly, $T_S=(V_S,E_S)$ constructed by some shortest path tree algorithm has energy cost
$\Theta(\vert U\vert^2)$, where $V_S=V$ and $E_S=\{(u_2,u_1),(u_1,r)\} \cup
\{(s_{i,1},r),(u_i,s_{i,i-2}) \vert 3 \le i \le \vert U\vert \} \cup \{(s_{i,j+1},s_{i,j}) \vert 3
\le i \le \vert U\vert, 1 \le j \le i-3\}$. However, $T_O=(V_O,E_O)$ has energy cost $\Theta(\vert
U\vert)$, where $V_O=\{r\} \cup U$ and $E_O=\{(u_1,r)\} \cup \{(u_{i+1},u_i) \vert 1 \le i \le
\vert U\vert-1\}$.
\end{proof}

\begin{figure}
\center \subfigure[]{\includegraphics[scale=.65]{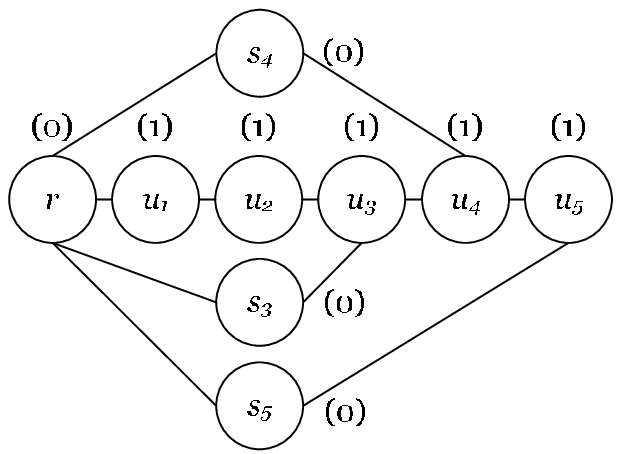}\label{Fig: Steiner tree instance}} \hfil
\subfigure[]{\includegraphics[scale=.65]{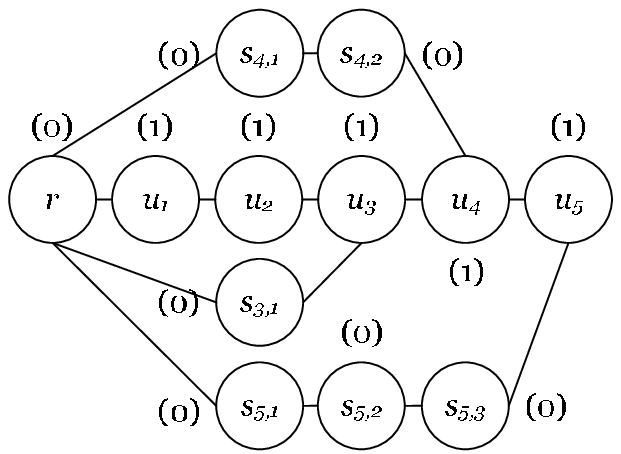}\label{Fig: shortest path tree instance}}\\
\caption{Instances of the MECAT$\_$RN problem. (a) An instance with the approximation ratio of a
Steiner tree algorithm $\Theta(\vert U \vert)$. (b) An instance with the approximation ratio of a
shortest path tree algorithm $\Theta(\vert U \vert)$.}
\end{figure}

\bibliographystyle{IEEEtran}
\bibliography{ARP}

% Generated by IEEEtran.bst, version: 1.13 (2008/09/30)
\begin{thebibliography}{10}
\providecommand{\url}[1]{#1}
\csname url@samestyle\endcsname
\providecommand{\newblock}{\relax}
\providecommand{\bibinfo}[2]{#2}
\providecommand{\BIBentrySTDinterwordspacing}{\spaceskip=0pt\relax}
\providecommand{\BIBentryALTinterwordstretchfactor}{4}
\providecommand{\BIBentryALTinterwordspacing}{\spaceskip=\fontdimen2\font plus
\BIBentryALTinterwordstretchfactor\fontdimen3\font minus
  \fontdimen4\font\relax}
\providecommand{\BIBforeignlanguage}[2]{{%
\expandafter\ifx\csname l@#1\endcsname\relax
\typeout{** WARNING: IEEEtran.bst: No hyphenation pattern has been}%
\typeout{** loaded for the language `#1'. Using the pattern for}%
\typeout{** the default language instead.}%
\else
\language=\csname l@#1\endcsname
\fi
#2}}
\providecommand{\BIBdecl}{\relax}
\BIBdecl

\bibitem{R1}
C.~Liu and G.~Cao, ``Distributed monitoring and aggregation in wireless sensor
  networks,'' in \emph{IEEE INFOCOM}, 2010.

\bibitem{R9}
A.~Mainwaring, D.~Culler, J.~Polastre, R.~Szewczyk, and J.~Anderson, ``Wireless
  sensor networks for habitat monitoring,'' in \emph{ACM WSNA}, 2002.

\bibitem{R10}
N.~Xu, S.~Rangwala, K.~K. Chintalapudi, D.~Ganesan, A.~Broad, R.~Govindan, and
  D.~Estrin, ``A wireless sensor network for structural monitoring,'' in
  \emph{ACM SenSys}, 2004.

\bibitem{R11}
A.~Giridhar and P.~R. Kumar, ``Computing and communicating functions over
  sensor networks,'' \emph{IEEE Journal on Selected Areas in Communications},
  vol.~23, pp. 755--764, 2005.

\bibitem{R12}
R.~Cristescu, B.~Beferull-Lozano, and M.~Vetterli, ``On network correlated data
  gathering,'' in \emph{IEEE INFOCOM}, 2004.

\bibitem{R13}
J.~Li, A.~Deshpande, and S.~Khuller, ``On computing compression trees for data
  collection in wireless sensor networks,'' in \emph{IEEE INFOCOM}, 2010.

\bibitem{R2}
D.~Luo, X.~Zhu, X.~Wu, and G.~Chen, ``Maximizing lifetime for the shortest path
  aggregation tree in wireless sensor networks,'' in \emph{IEEE INFOCOM}, 2011.

\bibitem{R5}
Y.~Wu, S.~Fahmy, and N.~B. Shroff, ``On the construction of a maximum-lifetime
  data gathering tree in sensor networks: {NP-Completeness} and approximation
  algorithm,'' in \emph{IEEE INFOCOM}, 2008.

\bibitem{R14}
A.~Goel and D.~Estrin, ``Simultaneous optimization for concave costs: {Single}
  sink aggregation or single source buy-at-bulk,'' in \emph{SODA}, 2003.

\bibitem{R23}
M.~Enachescu, A.~Goel, R.~Govindan, and R.~Motwani, ``Scale free aggregation in
  sensor networks,'' in \emph{ALGOSENSORS}, 2004.

\bibitem{R24}
N.~Thepvilojanapong, Y.~Tobe, and K.~Sezaki, ``On the construction of efficient
  data gathering tree in wireless sensor networks,'' in \emph{IEEE ISCAS},
  2005.

\bibitem{R27}
B.~Yu, J.~Li, and Y.~Li, ``Distributed data aggregation scheduling in wireless
  sensor networks,'' in \emph{IEEE INFOCOM}, 2009.

\bibitem{Sinha}
A.~Sinha and A.~Chandrakasan, ``Dynamic power management in wireless sensor
  networks,'' \emph{Design Test of Computers, IEEE}, vol.~18, pp. 62--74, 2001.

\bibitem{Min}
R.~Min, M.~Bhardwaj, S.-H. Cho, A.~Sinha, E.~Shih, A.~Wang, and
  A.~Chandrakasan, ``An architecture for a power-aware distributed microsensor
  node,'' in \emph{IEEE {SiPS}}, 2000.

\bibitem{octopus}
\BIBentryALTinterwordspacing
Octopus wireless sensor network. [Online]. Available:
  \url{http://163.13.128.59/}
\BIBentrySTDinterwordspacing

\bibitem{R20}
R.~Hassin, R.~Ravi, and F.~S. Salman, ``Approximation algorithms for a
  capacitated network design problem,'' \emph{Algorithmica}, vol.~38, pp.
  417--431, 2004.

\bibitem{Min02topfive}
R.~Min and A.~Chandrakasan, ``Top five myths about the energy consumption of
  wireless communication,'' \emph{ACM Sigmobile Mobile Communication and
  Communications Review}, vol.~6, pp. 65--67, 2002.

\bibitem{Wang:2001:EEM:383082.383105}
A.~Wang, S.~Cho, C.~Sodini, and A.~Chandrakasan, ``Energy efficient modulation
  and mac for asymmetric rf microsensor systems,'' in \emph{ACM ISLPED}, 2001.

\bibitem{R6}
C.~P. Low, ``An approximation algorithm for the load-balanced semi-matching
  problem in weighted bipartite graphs,'' \emph{Information Processing
  Letters}, vol. 100, pp. 154 -- 161, 2006.

\bibitem{R35}
M.-J. Tsai, H.-Y. Yang, and W.-Q. Huang, ``Axis-based virtual coordinate
  assignment protocol and delivery-guaranteed routing protocol in wireless
  sensor networks,'' in \emph{IEEE INFOCOM}, 2007.

\bibitem{R16}
X.~Cheng, D.-Z. Du, L.~Wang, and B.~Xu, ``Relay sensor placement in wireless
  sensor networks,'' \emph{Wireless Networks}, vol.~14, pp. 347--355, 2008.

\bibitem{R17}
X.~Han, X.~Cao, E.~L. Lloyd, and C.-C. Shen, ``Fault-tolerant relay node
  placement in heterogeneous wireless sensor networks,'' in \emph{IEEE
  INFOCOM}, 2007.

\bibitem{R18}
S.~Misra, S.~D. Hong, G.~Xue, and J.~Tang, ``Constrained relay node placement
  in wireless sensor networks to meet connectivity and survivability
  requirements,'' in \emph{IEEE INFOCOM}, 2008.

\bibitem{R7}
M.~R. Garey and D.~S. Johnson, \emph{Computers and Intractability: A Guide to
  the Theory of NP-Completeness}.\hskip 1em plus 0.5em minus 0.4em\relax New
  York, NY, USA: W. H. Freeman, 1979.

\bibitem{R21}
F.~S. Salman, J.~Cheriyan, R.~Ravi, and S.~Subramanian, ``Approximating the
  single-sink link-installation problem in network design,'' \emph{SIAM J. on
  Optimization}, vol.~11, pp. 595--610, 2000.

\bibitem{R22}
S.~Khuller, B.~Raghavachari, and N.~Young, ``Balancing minimum spanning trees
  and shortest-path trees,'' \emph{Algorithmica}, vol.~14, pp. 305--321, 1995.

\bibitem{Byrka}
J.~Byrka, F.~Grandoni, T.~Rothvo\ss, and L.~Sanit\`{a}, ``An improved lp-based
  approximation for steiner tree,'' in \emph{STOC}, 2010.

\bibitem{R34}
L.~Kou, G.~Markowsky, and L.~Berman, ``A fast algorithm for steiner trees,''
  \emph{Acta Informatica}, vol.~15, pp. 141--145, 1981.

\bibitem{R33}
C.~Intanagonwiwat, R.~Govindan, and D.~Estrin, ``Directed diffusion: {A}
  scalable and robust communication paradigm for sensor networks,'' in
  \emph{ACM MobiCom}, 2000.

\end{thebibliography}

\end{document}